\documentclass[USenglish]{article}
\usepackage{scrextend}
\usepackage{fullpage}
\usepackage{microtype}
\usepackage{hyperref}
\usepackage{amsthm,todonotes}

\usepackage[cmex10]{amsmath}
\usepackage{amssymb}
\usepackage{enumitem}
\usepackage{color}

\newtheorem{proposition}{Proposition}

\newtheorem{theorem}{Theorem}
\newtheorem{lemma}{Lemma}

\newtheorem{corollary}{Corollary}

\newcommand{\size}[1]{\mathsf{size}(#1)}

\newcommand{\VC}{\mathsf{VC}}
\newcommand{\calC}{C}
\newcommand{\calD}{D}
\newcommand{\calP}{\mathcal{P}}

\newcommand{\neigh}{N}

\newcommand{\node}{\text{gate}}
\newcommand{\nodes}{\text{gates}}
\newcommand{\arc}{\text{wire}}
\newcommand{\arcs}{\text{wires}}
\newcommand{\indegree}{\text{fanin}}

\newcommand{\inarc}{\text{input wire}}

\newcommand{\inarcs}{\text{input wires}}

\newcommand{\source}{\text{input}}

\newcommand{\sink}{\text{output}}

\newcommand{\sub}{\mathsf{sub}}
\newcommand{\cert}{\mathsf{cert}}
\newcommand{\vars}{\mathsf{vars}}

\newcommand{\outgate}[1]{\mathsf{output}(#1)}
\newcommand{\ingates}[1]{\mathsf{inputs}(#1)}

\begin{document}

\title{A Strongly Exponential Separation\\of DNNFs from CNF Formulas}

\newcommand*\samethanks[1][\value{footnote}]{\footnotemark[#1]}
\author{Simone Bova\thanks{Vienna University of Technology} \and Florent Capelli\thanks{IMJ UMR 7586  -  Logique,
Universit\'{e} Paris Diderot,
France} \and Stefan Mengel\thanks{ 
LIX UMR 7161,
Ecole Polytechnique,
France}\and Friedrich Slivovsky\samethanks[1]}

\date{}

\maketitle

\begin{abstract}
Decomposable Negation Normal Forms (DNNFs) are Boolean circuits in negation normal form where the subcircuits leading into each AND gate are defined on disjoint sets of variables. We prove a strongly exponential lower bound on the size of DNNFs for a class of CNF formulas built from expander graphs. As a corollary, we obtain a strongly exponential separation between DNNFs and CNF formulas in prime implicates form. This settles an open problem in the area of knowledge compilation (Darwiche and Marquis, 2002).
\end{abstract}

\section{Introduction}
The aim of knowledge compilation is to succinctly represent propositional knowledge bases in a format that allows for answering a number of queries in polynomial time~\cite{DarwicheM02}. 
Choosing a representation language generally involves a trade-off between succinctness and the range of queries that can be efficiently answered. For instance, CNF formulas are more succinct than prime implicate formulas (PIs), but the latter representation enjoys clause entailment checks in polynomial time whereas CNF formulas in general do not, unless $\mathrm{P}=\mathrm{NP}$~\cite{GogicKPS95,Chandra78}. The need to balance the competing requirements of succinctness and tractability has led to the introduction of a large variety of representation languages that strike this balance in different ways.

Decomposable Negation Normal Forms (DNNFs) are Boolean circuits in negation normal form (NNF) such that the subcircuits leading into an AND gate are defined on disjoint sets of variables~\cite{Darwiche01}. 
DNNFs are among the most succinct representation languages considered in knowledge compilation---for instance, they generalize variants of binary decision diagrams such as ordered binary decision diagrams (OBDDs) and even free binary decision diagrams (FBDDs, also known as read-once branching programs). They have also been studied in circuit complexity, 
under the name of multilinear Boolean circuits~\cite{SenguptaV94,PonnuswamiV04,Krieger07}. 

In this paper, we consider the relative succinctness of DNNFs and CNF formulas. On the one hand, DNNFs can be exponentially more succinct than CNF formulas~\cite{GogicKPS95}. On the other hand, Darwiche and Marquis observed that CNFs do not admit polynomial DNNF representations unless the polynomial hierarchy collapses, while posing an unconditional proof of such a separation as an open problem~\cite{DarwicheM02}. 

An unconditional, \emph{weakly exponential} separation can be derived from known results (see the section on related work below). By using a more direct construction that leverages the combinatorial properties of expander graphs, 
we obtain a \emph{strongly exponential} separation (Theorem~\ref{thm:mainnontechnical}):
\begin{quote}
\emph{There is a class $\mathcal{C}$ of CNF formulas such that for each $F \in \mathcal{C}$, the DNNF size of $F$ is $2^{\Omega(n)}$, where $n$ is the number of variables of $F$.}
\end{quote}
The formulas in $\mathcal{C}$ satisfy strong syntactic restrictions. In particular, they are in prime implicate form, so we immediately obtain an exponential separation of DNNFs from PIs~(Corollary~\ref{cor:pi}), answering an open question by Darwiche and Marquis~\cite{DarwicheM02}.

Our result further improves the best known lower bound of $\Omega\big(\frac{2^{\sqrt{n}}}{\sqrt[4]{n}}\big)$ on the DNNF size of a Boolean function of $n$ variables~\cite{Krieger07}.

\paragraph*{Related Work.} We observe that an unconditional, weakly exponential separation of DNNFs from CNF formulas can be obtained from known results as follows. Let $F$ be a CNF formula encoding the run of a nondeterministic polynomial-time Turing machine deciding the clique problem for some fixed input size. A DNNF representation of~$F$ can be turned into a DNNF computing the clique function by projecting on the variables encoding the input. This can be done without increasing the size of the DNNF~\cite{Darwiche01}. Since an optimal DNNF computing a monotone function is monotone \cite{Krieger07}, weakly exponential lower bounds for monotone circuits computing the clique function~\cite{AlonB87} transfer to lower bounds on the DNNF size of~$F$.

The formulas used to prove our main result are based on expander graphs and were originally introduced 
to establish an exponential lower bound for the OBDD size of CNFs~\cite{BovaS14}. The present paper leverages 
a recent result by Razgon~\cite[Theorem~4]{Razgon14c} to lift this lower bound to DNNFs; indeed, 
we slightly improve Razgon's result, while at the same time providing a significantly shorter proof 
based on a new combinatorial result (Theorem~\ref{th:florent}).

There is a rich literature on lower bounds for more restricted representation languages, such as OBDDs or FBDDs~\cite{Wegener00,Jukna12}. Moreover, certain subclasses of DNNFs, so-called decision-DNNFs, have been recently considered in database theory in the context of probabilistic databases. In this setting, lower bounds are obtained by a quasipolynomial simulation of decision-DNNFs by FBDDs in combination with known exponential lower bounds for FBDDs~\cite{BeameLRS13,BeameLRS14}. 
Pipatsrisawat and Darwiche have proposed a framework for showing lower bounds on structured DNNFs~\cite{PipatsrisawatD10}, a subclass of DNNFs in which the variables respect a common tree-ordering. 

\section{Preliminaries}\label{sect:prelim}

Let $X$ be a countable set of variables.  A literal is a variable ($x$) or a negated variable ($\neg x$).  
An assignment is a function $f$ from $X$ to the constants $0$ and $1$.  We occasionally  
identify an assignment $f$ with the set of literals $\{ \neg x \colon f(x)=0 \} \cup \{ x \colon f(x)=1 \}$.  

\paragraph*{NNFs.} A \emph{negation normal form (NNF)} $C$ (also known as a \emph{De Morgan} circuit) 
is a node labeled directed acyclic graph (DAG), 
whose labeled nodes and arcs are respectively called the \emph{gates} and \emph{wires} of $C$.  
The underlying DAG has a unique sink (outdegree $0$) node, referred to as the \emph{output gate} of the circuit, and denoted by $\outgate{C}$.  
The source nodes of $C$ (indegree $0$), denoted by $\ingates{C}$, are referred to as the \emph{input gates} of $C$ and are labeled by a constant 
($0$ or $1$) or by a literal $x$ or $\neg x$ for $x \in X$.  We let $\mathsf{vars}(C)$ 
denote the set of variables occurring in the labels of input gates of $C$.  The non source nodes of $C$, referred as \emph{internal gates}, 
are labeled by $\wedge$ or $\vee$.  
In this paper, the \emph{size} of $C$, in symbols $\mathsf{size}(C)$, is the number of wires in $C$.  

Let $G$ be a DAG and let $v$ be a node in $G$. The \emph{subgraph of $G$ sinked at $v$} 
is the DAG whose node set is 
$$V'=\{v\} \cup \{ u \colon \text{there exists a directed path from $u$ to $v$ in $G$} \}\text{,}$$
and whose arcs are exactly the arcs of $G$ among the nodes in $V'$.  Let $C$ be an NNF, 
and let $v$ be a node in the DAG $G$ underlying $C$.  We let $\mathsf{sub}(C,v)$ denote the 
\emph{subcircuit of $C$ sinked at $v$}, that is, the NNF 
whose underlying DAG is the subgraph of $G$ sinked at $v$, with the same labels, 
and whose variables are those labeling the input gates of $\mathsf{sub}(C,v)$.

Let $C$ be an NNF, $v$ be a \node\ in $C$, and $f$ be an assignment.  
The \emph{value of $\mathsf{sub}(C,v)$ under $f$}, in symbols $\mathsf{sub}(C,v)(f)$, 
is defined inductively as usual.  The \emph{value of $C$ under $f$}, in symbols $C(f)$, 
is equal to $\mathsf{sub}(C,\outgate{C})(f)$.  We let $\mathsf{sat}(C)$ denote the set of \emph{satisfying assignments} of $C$, 
that is, assignments $f$ such that $C(f)=1$.  
Let $C$ be an NNF on variables $\mathsf{vars}(C)$.  Two NNFs $C$ and $C'$ are \emph{equivalent} 
if $$C(f)=C'(f)$$ for all assignments $f$.

\paragraph*{Certificates.} Let $C$ be an NNF.  
A \emph{certificate} of $C$ 
is an NNF $T$ 
whose \nodes\ and \arcs\ are subsets of the \nodes\ and \arcs\ of $C$ satisfying the following: 
\begin{itemize}
\item $\outgate{C}=\outgate{T}$; 
\item if a \node\ $v$ is in $C \cap T$, and $v$ is a $\wedge$-\node\ in $C$ with \inarcs\ from \nodes\ $v_1,\ldots,v_i$, 
then the \nodes\ $v_1,\ldots,v_i$ and the \arcs\ $(v_1,v),\ldots,(v_i,v)$ are in $T$ ($i \geq 0$); 
\item if a \node\ $v$ is in $C \cap T$, and $v$ is an $\vee$-\node\ in $C$ with \inarcs\ from \nodes\ $v_1,\ldots,v_i$, 
then exactly one \node\ $w \in \{v_1,\ldots,v_i\}$ is in $T$ and the \arc\ $(w,v)$ is in $T$ ($i \geq 0$).
\end{itemize}
We let $\mathsf{cert}(C)$ denote the set of certificates of $C$.  

Satisfying assignments and certificates of an NNF are nicely related as follows.  

\begin{proposition}
\label{prop:mods-vs-trees}
Let $C$ be an NNF and let $f$ be an assignment.  
Then, $f \in \mathsf{sat}(C)$ if and only if there exists $T \in \mathsf{cert}(C)$ such that $f \in \mathsf{sat}(T)$.
\end{proposition}

\newcommand{\pfmodsvstrees}[0]{
\begin{proof}
For the forward direction, let $f \in \mathsf{sat}(C)$.  
Call a \arc\ $(u,v)$ in the DAG underlying $C$ \emph{activated} by $f$ 
if $f$ satisfies the subcircuit of $C$ sinked at $u$, in symbols $\mathsf{sub}(C,u)(f)=1$.  
It is readily verified that there exists a certificate $T$ for $C$ 
containing only \arcs\ activated by $f$.  Moreover, $f \in \mathsf{sat}(T)$ 
because by construction $\mathsf{sub}(T,v)(f)=1$ for all \source\ \nodes\ $v$ of $T$, 
and therefore $\mathsf{sub}(T,\outgate{T})(f)=1$.

For the backward direction, let $T$ be a certificate of $C$ such that 
$f \in \mathsf{sat}(T)$.  By induction on the structure of $C$, we prove that for all $t \in T$ it holds that 
$\mathsf{sub}(T,t)(f)=\mathsf{sub}(C,t)(f)=1$.  Since $\outgate{C}=\outgate{T}$, 
we conclude that $$C(f)=\mathsf{sub}(C,\outgate{C})(f)=\mathsf{sub}(T,\outgate{T})(f)=1\text{.}$$  Then, 
$f \in \mathsf{sat}(C)$.

If $t$ is an \source\ \node\ of $T$, then $\mathsf{sub}(T,t)(f)=1$ because 
otherwise $f$ does not satisfy $T$.  We also have that $t$ is an \source\ \node\ of $C$, 
hence $\mathsf{sub}(C,t)(f)=1$.  Let $t$ be a $\vee$-\node\ in $T$, 
with \inarcs\ from \nodes\ $t_1,\ldots,t_i$ in $C$; 
say without loss of generality that $t_1$ is chosen in $T$.  
By the induction hypothesis, $\mathsf{sub}(T,t_1)(f)=\mathsf{sub}(C,t_1)(f)=1$; 
hence, $\mathsf{sub}(T,t)(f)=\mathsf{sub}(C,t)(f)=1$.  The case where 
$t$ is a $\wedge$-\node\ is similar.
\end{proof}}

\pfmodsvstrees

\paragraph*{DNNFs and CNFs.} An NNF $D$ is \emph{decomposable} (in short, a \emph{DNNF}) if for all $\wedge$-\nodes\ $v$ 
with \inarcs\ from \nodes\ $v_1,\ldots,v_i$ and all $j,j' \in \{1,\ldots,i\}$, $j \neq j'$, the variable sets of 
the subcircuits of $D$ sinked at $v_j$ and $v_{j'}$ are disjoint, in symbols, 
$$\mathsf{vars}(\mathsf{sub}(D,v_j)) \cap \mathsf{vars}(\mathsf{sub}(D,v_{j'})) = \emptyset\text{.}$$

A \emph{conjunctive normal form} (in short, \emph{CNF}) is a finite conjunction 
of clauses (finite disjunctions of literals).  Equivalently, a CNF is 
an NNF where the maximum number of \arcs\ on a path from a \source\ \node\ to the \sink\ \node\ is $2$, 
and each $\vee$-\node\ has \inarcs\ only from \source\ \nodes.  A CNF $F$ is \emph{monotone} if 
its labels do not contain negative literals, a \emph{$k$-CNF} if the \indegree\ of $\vee$-\nodes\ is at most $k$, 
and a \emph{read $k$ times} CNF if, for every $x \in \mathsf{vars}(F)$, 
the number of \arcs\ leaving nodes whose label contain the variable $x$ is at most $k$.

For a CNF $E$, we denote by $\mathsf{DNNF}(E)$ the size of the smallest DNNF equivalent to $E$, that is
$$\mathsf{DNNF}(E)=\min\{ \mathsf{size}(D) \colon \text{$D$ is a DNNF equivalent to $E$} \}\text{.}$$

\paragraph*{Graphs.}  We refer to a standard reference for basic notions and facts in graph theory \cite{Diestel05}.  
Let $G=(V,E)$ be a graph.  A \emph{vertex cover} of $G$ is a subset $C$ of the vertices $V$ such that $\{u,v\} \cap C \neq \emptyset$ for all $\{u,v\} \in E$.  We denote by $\VC(G)$ the set of the vertex covers of $G$.  

We observe two facts about graphs that will be useful later in proving the main result.  
The first is that vertex covers of graphs in a class of graphs of bounded degree are large.

\begin{proposition}
\label{prop:vcbounds}
Let $C \in \VC(G)$ be a vertex cover of a connected graph $G=(V,E)$ of degree $d$.  
Then, $|V|/(d+1) \leq |C|$.  
\end{proposition}

\newcommand{\pfvcbounds}[0]{
\begin{proof}
Let $C$ be a vertex cover of a connected graph $G$ of maximum degree $d$.  
Then, $V \setminus C$ is an independent set.  Since $G$ is connected, 
each vertex in $V \setminus C$ is incident 
to at least one edge with a vertex in $C$.  Hence, 
there are at least $|V\setminus C|$ edges between $C$ 
and $V\setminus C$.  Since each vertex in $C$ has degree at most $d$, 
$$|V|-|C|=|V\setminus C| \leq d|C|\text{,}$$  
and we are done.  
\end{proof}}

\pfvcbounds

The second is that a rooted binary tree with a large number of leaves 
always contains a subtree with a large but not too large number of leaves.  

\begin{proposition}
\label{prop:tree}
Let $T$ be a rooted binary tree with at least $\ell$ leaves. Then there exists a vertex $v$ of $T$ such that the number of leaves of the subtree of $T$ rooted in $v$ is at least $\ell$ and at most $2\ell$.
\end{proposition}

\newcommand{\pftree}[0]{
\begin{proof}
  The proof is by induction on the size of $T$. If the number of leaves of $T$ is already between $\ell$ and $2\ell$, then we can choose the root. Now if $T$ has more than $2 \ell$ leaves then either the root has one child $w$. In this case, we apply the induction hypothesis on the subtree of $T$ rooted in $w$, since it has also more than $2 \ell$ leaves, that is more than $\ell$ leaves.

  Now assume that the root has two children $w_1$, $w_2$. Let $T_1$ and $T_2$ be the subtrees rooted in $w_1$ and $w_2$ respectively. Assume without loss of generality that $T_1$ has more leaves than $T_2$. Thus $T_1$ has more than $2\ell/2=\ell$ leaves. By induction, there exists a vertex $v$ in $T_1$ and thus in $T$ such that the subtree rooted in $v$ has at least $\ell$ and at most $2\ell$ leaves.
\end{proof}
}
\pftree

\paragraph*{Expanders.} Let $G=(V,E)$ be a graph. %($V \neq \emptyset$).  
For every $S \subseteq V$, we let $\neigh_{S,G}$ 
denote the \emph{open neighbourhood} of $S$ in $G$, in symbols, 
$$\neigh_{S,G}=\{ v \in V \setminus S \colon \textup{ there exists $u \in S$ such that $\{u,v\} \in E$} \}\text{;}$$
we write $\neigh_v$ instead of $\neigh_{\{v\},G}$ ($v \in V$), 
and $\neigh_S$ instead of $\neigh_{S,G}$ if the intended graph $G$ is clear from the context.  
For a vertex $v \in V$, we denote the degree of $v$ by $d(v) = |\neigh_v|$; 
the degree of $G$ is the maximum degree attained over its vertices.

Let $d \geq 3$ and $c>0$.  A graph $G=(V,E)$ is a \emph{$(c,d)$-expander} 
if $G$ has degree~$d$ and for all $S \subseteq V$ such that $|S| \leq |V|/2$ it holds that 
\begin{equation}\label{eq:expansion}
|\neigh_{S,G}| \geq c|S| \text{.}
\end{equation}
Note that a $(c,d)$-expander is connected and that taking $|S|=|V|/2$ implies that $c \leq 1$.

\begin{theorem}[Section 9.2 in \cite{AlonSpencer00}]\label{th:expex}
For all $d \geq 3$, there exists $c>0$ and a sequence 
of graphs $\{G_i \mid i \in \mathbb{N} \}$ such that $G_i=(V_i,E_i)$ is a $(c,d)$-expander  
and $|V_i| \to \infty$ as $i  \to \infty$ ($i \in \mathbb{N}$). 
\end{theorem}

\newcommand{\pfflorent}[0]{
\begin{proof}
Let $G = (V,E)$ be a graph of degree $d$, and let $S \subseteq V$.  By induction on $|S| \geq 0$, 
we prove that 
$$|\VC(G, S)| \leq \bigg(\prod_{s \in S} \frac{2^{d(s)}}{1+2^{d(s)}}\bigg)|\VC(G)|\text{.}$$
The statement follows since for all $s \in S$, we have $d(s) \leq d$ and thus:
$$\frac{2^{d(s)}}{1+2^{d(s)}} \leq \frac{2^{d}}{1+2^{d}}$$ 

The base case $S = \emptyset$ is trivial. If $S \neq \emptyset$, let $t \in S$. By Theorem~\ref{th:florent}, 
\[ |\VC(G,S)|  \leq \frac{2^{d(t)}}{1+2^{d(t)}} |\VC(G, S\setminus \{t\})|\]
and thus by the induction hypothesis
\[ |\VC(G,S)|  \leq \frac{2^{d(t)}}{1+2^{d(t)}} \bigg(\prod_{s \in S \setminus \{t\}} \frac{2^{d(s)}}{1+2^{d(s)}}\bigg) |\VC(G)| = \bigg(\prod_{s \in S} \frac{2^{d(s)}}{1+2^{d(s)}}\bigg) |\VC(G)|\]
\end{proof}}

\newcommand{\pfkeyflorent}[0]{
\begin{proof}%[Proof of Theorem~\ref{th:florent}]
Let $G = (V,E)$ be a graph, $S \subseteq V$, and $s \in S$. Let $f$ be the mapping defined as $f(\calC) = (\neigh_s \cap \calC, (\calC \setminus \{s\}) \cup \neigh_s)$ for all $\calC \in \VC(G,S)$. We denote by $\calP(N_s)$ the power set of $N_s$, that is $\{A : A \subseteq N_s\}$.

First remark that for all $\calC \in \VC(G,S)$, $f(\calC) = (A, \calD) \in \calP(\neigh_s) \times \big ( \VC(G, S\setminus \{s\}) \setminus \VC(G,S) \big )$. It is clear that $A = \calC \cap \neigh_s \subseteq \neigh_s$ thus $A \in \calP(\neigh_s)$. Moreover, if $\calC$ is a vertex cover of $G$, then $\calD = (\calC \setminus \{s\}) \cup \neigh_s$ is also a vertex cover of~$G$, since each edge $e$ of $G$ is covered by $\calC$: If $s$ is not an endpoint of $e$, then $e$ is still covered by $\calC \setminus \{s\}$, and thus also by $\calD$. Otherwise $e = \{s,t\}$ with $t \in \neigh_s$. Thus $e$ is covered by $\calD$ since $t \in \calD$. Finally, if $S \subseteq \calC$, then $S\setminus\{s\} \subseteq \calD$ and $s \notin \calD$. Thus $\calD \in \VC(G, S \setminus \{s\})$ and $\calD \notin \VC(G, S)$.

We now prove that $f$ is an injection. Let $\calC, \calC' \in \VC(G,S)$ such that $(A,\calD) = f(\calC) = f(\calC')$. Then, by definition, $\calC \cap \neigh_s = \calC' \cap \neigh_s$ and $\calC \setminus \{s\} \cup \neigh_s = \calC' \setminus \{s\} \cup \neigh_s$. Since $s$ is both in $\calC$ and in $\calC'$, we have $\calC \cup \neigh_s = \calC' \cup \neigh_s$ and $\calC \cap \neigh_s = \calC' \cap \neigh_s$, that is $C = C'$. Thus $f$ is an injection. It follows that:
\[ |\VC(G,S)| \leq |\calP(\neigh_s)| \times |\VC(G,S\setminus\{s\}) \setminus \VC(G,S)|.\]
Since $\VC(G,S) \subseteq \VC(G,S\setminus\{s\})$ , we have $|\VC(G,S\setminus\{s\}) \setminus \VC(G,S)| = |\VC(G,S\setminus\{s\})|-|\VC(G,S)|$ and it is clear that $|\calP(\neigh_s)| = 2^{d(s)}$. It follows that:
$${|\VC(G, S)|} \leq \left(\frac{2^{d(s)}}{1+2^{d(s)}}\right){|\VC(G,S\setminus \{s\})|}\text{.}$$
\end{proof}}

\section{Outline of the Proof}

We consider a class of what we call graph CNFs.  
A \emph{graph CNF} is a monotone $2$-CNF corresponding to a graph  
in that a clause $x \vee y$ in the CNF corresponds to an edge $\{x,y\}$ in the graph.  Note that 
the models of a graph CNF correspond exactly to the vertex covers of the underlying graph. % (see Subsection~\ref{ssect:graphcnf}).  

More specifically, our graph CNFs correspond to an infinite family of \emph{expander graphs}.  
A graph $G=(V,E)$ in this family is highly connected but sparse 
(in the sense of Theorem~\ref{th:expex}).  As a consequence, 
given any $S \subseteq V$ of size no larger than $|V|/2$, but linear in $|V|$, 
it is possible to find a matching of size linear in $|V|$ between $S$ and $V\setminus S$ (Corollary~\ref{cor:large-point-large-matching}).  
This is crucial to establish a strongly exponential lower bound.

%RELATED WORK which then immediately lift to strong exponential lower bounds in the length of the CNF 
% because the number of edges in a constant degree graph is linear in the number of vertices (Theorem~\ref{th:maintechnical}).  
% This combinatorial property of expanders has been exploited (not to its full extent indeed) 
% to give strong exponential lower bounds on the OBDD size of expander graph CNFs \cite{BovaS15}; 
% here we prove that it yields strong exponential lower bounds in the more general setting of DNNFs.

An optimal DNNF computing a monotone Boolean function is monotone \cite[Lemma~3]{Krieger07}.  
Since graph CNFs are monotone, it suffices to prove a lower bound for monotone DNNFs (Proposition~\ref{prop:normaldnnf}).  
We do this by means of a \emph{bottleneck counting} argument \cite{Haken85}: we identify a set $B$ of gates such that each satisfying assignment of the DNNF has to pass through one of these gates, and argue that the number of assignments passing through an individual gate is small. Since the number of satisfying assignments is large, we conclude that the number of gates must be large as well.

More specifically, the argument goes as follows.
Let $F$ be a graph CNF whose underlying graph $G=(V,E)$ is an expander (of degree $d$), 
and let $D$ be a (nice) DNNF computing $F$. The set $B$ of gates is defined by taking, for every certificate $T$ of $D$, 
a gate $v_T$ in $T$ such that the number of variables in the subcircuit of $D$ rooted at $v_T$ 
is between $|V|/(d+1)$ and $|V|/2$ (Lemma~\ref{lemma:gates}).  
It follows from the expansion properties of $G$ (Lemma~\ref{lemma:point}) 
and the decomposability properties of $D$ (Theorem~\ref{th:treesagreeonedges} and Corollary~\ref{cor:large-point-large-matching}) 
that for all gates $v \in B$ there exists a subset $I_v$ of $V$ of size linear in $|V|$ 
such that, for all certificates $T$ of $D$ containing the gate~$v$, 
it holds that $I_v$ is contained in the variables of $T$. The satisfying assignments of $D$ passing through $v$ are those mapping all variables in $I_v$ to $1$.

Next, we show that for every $v \in B$, 
the fraction of satisfying assignments of $D$ containing $I_v$ %(that is, $l^{-1}(v)$) 
is exponentially small in $|V|$ (Theorem~\ref{th:florent} and Corollary~\ref{cor:keyrazgon}). Moreover, the union (over gates $v \in B$) of satisfying assignments of $D$ mapping $I_v$ to $1$ coincides with the satisfying assignments of $D$.  
It follows that the size of $B$ is exponentially large in $|V|$ (Theorem~\ref{th:maintechnical}).

\section{Proof of the Lower Bound}

In this section, we prove our main result.  We introduce graph CNFs and nice DNNFs, 
prove a key property of nice DNNFs computing graph CNFs (Section~\ref{sect:graphcnfs}), 
and present our bottleneck argument (Section~\ref{sect:bottleneck}).  

\subsection{Graph CNFs and Nice DNNFs}\label{sect:graphcnfs}

If $G$ is a graph with at least two vertices and no isolated vertices, 
we view the edge set of $G$ as a CNF on the variables $\mathsf{vars}(E)=V$, namely,
\begin{equation}\label{eq:graphcnf}
\bigwedge_{\{x,y\} \in E}(x \vee y)\text{;} 
\end{equation}
we call a CNF of the from (\ref{eq:graphcnf}) a \emph{graph CNF}, 
and identify it with $E$.  

Note that the satisfying assignments of a graph CNF $E$ correspond 
to vertex covers of the underlying graph $G=(V,E)$ as follows: 
If $f$ is a satisfying assignment of $E$, then $\{x \in V \colon f(x)=1\} \in \VC(G)$, 
and if $V' \in \VC(G)$, then any assignment $f$ such that $V' \subseteq f$ satisfies~$E$.

An NNF is called \emph{negation free} 
if no input gate is labeled by a negated variable ($\neg x$), 
and \emph{constant free} if no input gate is labeled by a constant ($0$ or $1$).  
Note that, if $C$ is a negation and constant free NNF, 
then $\mathsf{vars}(C)$ coincides with the labels of the input gates of $C$.  
A \indegree\ $2$, constant free, and negation free DNNF is called \emph{nice}. 

The following statement implies that the minimum size of a nice DNNF computing 
a graph CNF (but indeed, more generally, any monotone Boolean function) is at most $2$ times as large as its DNNF size.

\begin{proposition}
\label{prop:normaldnnf}
Let $E$ be a graph CNF and let $D$ be a DNNF equivalent to~$E$.  
There exists a nice DNNF $D'$ equivalent to $D$ such that $\mathsf{size}(D') \leq 2 \cdot \mathsf{size}(D)$.
\end{proposition}

\newcommand{\normalformsection}[0]{
\newcommand{\pfindegtwo}[0]{
\begin{proof}
Let $D$ be a DNNF.  An NNF $D'$ equivalent to $D$ and having \indegree\ $2$ 
is obtained by editing $D$ as follows, until no \node\ of \indegree\ larger than $2$ exists: 
Let $v$ be a $\wedge$-\node\ with \inarcs\ from \nodes\ $v_1,\ldots,v_i$ with $i>2$; 
delete the \arcs\ $(v_j,v)$ for $j \in \{2,\ldots,i\}$; 
create a fresh $\wedge$-\node\ $w$, 
and the \arcs\ $(w,v)$ and $(v_j,w)$ for $j \in \{2,\ldots,i\}$.  
The case where $v$ is an $\vee$-\node\ is similar.  

It is readily verified that $D'$ is decomposable.  Moreover, 
each \arc\ in $D$ is processed at most once (when it is an \inarc\ of a \node\ having \indegree\ larger than $2$) 
and it generates at most $2$ \arcs\ in $D'$, hence the size of $D'$ is at most twice the size of $D$.
\end{proof}}

We first reduce to the \indegree\ $2$ case.

\begin{proposition}
\label{prop:indeg-two}
Let $D$ be a DNNF.  There exists a DNNF $D'$ equivalent to $D$, 
having \indegree\ $2$, and such that $\mathsf{size}(D') \leq 2 \cdot \mathsf{size}(D)$.
\end{proposition}

\pfindegtwo

Next, we reduce to the negation free case.  A Boolean function $F \colon \{0,1\}^{Y} \to \{0,1\}$ is called \emph{monotone} if 
for all assignments $f,f' \in Y \to \{0,1\}$ such that $f(x) \leq f'(x)$ for all $x \in Y$, it holds that $F(f) \leq F(f')$.

\newcommand{\pfnegfree}[0]{
\begin{proof}
Suppose $D$ contains a \node\ $u$ labeled with literal $\neg x$. 
Let~$D'$ be the DNNF on $\mathsf{vars}(D)$ obtained from~$D$ by relabeling $u$ with the constant $1$. 
We claim that an assignment satisfies $D$ if and only if it satisfies~$D'$. 
Let $f$ be a satisfying assignment of $D$. By Proposition~\ref{prop:mods-vs-trees} there is a certificate~$T$ 
of $D$ such that~$f$ satisfies $T$. We obtain a certificate $T'$ of $D'$ by relabeling $u$ with the constant $1$ (if $u$ appears in~$T$). 
It is straightforward to verify that $f$ is a satisfying assignment of $T'$. We apply Proposition~\ref{prop:mods-vs-trees} once more 
to conclude that $f$ must be a satisfying assignment of~$D'$. For the converse, let $f$ be a satisfying assignment of~$D'$, 
and let~$T'$ be a certificate of~$D'$ such that $f$ satisfies $T'$. If $T'$ does not contain the \node\ $u$ then $T'$ is also a certificate of $D$, 
and $f$ is a satisfying assignment of $D$ by Proposition~\ref{prop:mods-vs-trees}. Otherwise, we obtain a certificate~$T$ of $D$ from $T'$ 
by relabeling $u$ with the literal $\neg x$. Let $f'$ be the assignment such that $f(y) = f'(y)$ for 
all $y \in \mathsf{vars}(D) \setminus \{x\}$, and such that $f'(x) = 0$. Since $D$ is decomposable and $T$ contains the \node\ $u$ 
labeled with $\neg x$, no node of $T$ can be labeled with the literal~$x$. Thus $T'$ cannot contain such a \node\ either, 
and $f'$ satisfies $T'$. Since $\neg x$ evaluates to $1$ under~$f'$, the certificate $T$ is satisfied by $f'$ as well. 
By Proposition~\ref{prop:mods-vs-trees}, the assignment~$f'$ is a satisfying assignment of~$D$. 
Because the function computed by $D$ is monotone, we conclude that $f$ must satisfy $D$ as well. 
Clearly~$D$ and $D'$ have the same size and maximum \indegree. It follows that the desired negation free DNNF can be obtained by 
replacing every negative literal by the constant $1$ in the labels of $D$.
qed\end{proof}}

\begin{proposition}\cite[Lemma~3]{Krieger07}
\label{prop:neg-free}
Let $D$ be a DNNF computing a monotone Boolean function.  
There exists a DNNF $D'$ equivalent to $D$, 
negation free, and such that $\mathsf{size}(D') \leq \mathsf{size}(D)$. Moreover, $D'$ has the same \indegree\ as $D$.
\end{proposition}

\pfnegfree

Finally, we reduce to the constant free case.  Let $D$ be a DNNF not equivalent to $0$ or $1$.  
A constant free DNNF, denoted by $\mathsf{elimconst}(D)$, 
is obtained by editing $D$ as follows, until all \nodes\ labeled by a constant are deleted: 
Let $v$ be a $0$-\node, and let $v$ have \arcs\ to \nodes\ $v_1,\ldots,v_r$.  For all $j \in [r]$:  
if $v_j$ is a $\wedge$-\node, relabel $v_j$ by $0$, 
and delete all the \inarcs\ of $v_j$ 
(possibly creating some undesignated sink nodes in the underlying DAG); 
if $v_j$ is a $\vee$-\node, then delete the \arc\ $(v,v_j)$; 
relabel $v_j$ by $0$ if it becomes \indegree\ $0$; 
finally, delete $v$.  
The case where $v$ is a $1$-\node\ is similar.  Clearly, 

\begin{proposition}\label{prop:const-free}
Let $D$ be a non-constant DNNF and let $D'=\mathsf{elimconst}(D)$.  
Then, $D$ and $D'$ are equivalent, 
and $\mathsf{size}(D') \leq \mathsf{size}(D)$; 
moreover, the \indegree\ and negation freeness of $D$ are preserved in $D'$.
\end{proposition}

We conclude proving the statement.  

\begin{proof}[Proof of Proposition~\ref{prop:normaldnnf}]
Let $E$ be a graph CNF and let $D$ be a DNNF equivalent to $E$.  
By Proposition~\ref{prop:indeg-two}, there exists a DNNF $D_1$, equivalent to $D$, having \indegree\ $2$ 
whose size is at most twice the size of $D$.  
Since $E$ is a monotone Boolean function, by Proposition~\ref{prop:neg-free}, 
there exists a \indegree\ $2$ and negation free DNNF $D_2$, equivalent to $D_1$,  
whose size is at most the size of $D_1$.  
Since $E$ is a non constant Boolean function, by Proposition~\ref{prop:const-free}, 
there exists a \indegree\ $2$, negation free, and constant free DNNF $D_3$, equivalent to $D_2$, 
whose size is at most the size of $D_3$.  Let $D'=D_3$.  
Then, $D'$ is a \indegree\ $2$, constant free, and negation free DNNF, equivalent to $D$, 
whose size is at most twice the size of $D$.
\end{proof}}

\normalformsection

We also observe that, because of decomposability, certificates of nice DNNFs are tree shaped. 

\begin{proposition}
\label{prop:cert-are-trees}
Let $D$ be a (\indegree\ $2$) constant free DNNF and let $T \in \mathsf{cert}(D)$.  
The undirected graph underlying $T$ is a (binary) tree.  
Moreover, no two leaves of $T$ are labeled by the same variable.
\end{proposition}

\newcommand{\pfcertaretrees}[0]{
\begin{proof}
Assume that the undirected graph underlying $T$ is cyclic, 
so that in the underlying DAG there exist two distinct nodes $v$ and $w$ in $T$ 
and two arc disjoint directed paths from $v$ to $w$; 
in particular, $w$ has at least two ingoing arcs in $T$, 
hence by construction $w$ is a $\wedge$-\node\ in $D$.  By decomposability, 
no variables occur as labels of \source\ \nodes\ in $\mathsf{sub}(D,v)$, 
which is impossible since $D$ is constant free.  

We now consider that $T$ is rooted in $\mathsf{output}(D)$. Let $\ell_1$ and $\ell_2$ be two distinct leaves of $T$. 
Their least common ancestor $w$ in $T$ has two ingoing arcs, thus it is an $\land$-gate. 
By decomposability of $w$ in $D$, $\ell_1$ and $\ell_2$ are labeled by a different variable.
\end{proof}}

\pfcertaretrees

Let $D$ be a nice DNNF computing a graph CNF $E$ with underlying graph $G=(V,E)$, 
and let $\{x,x'\}$ be an edge (clause) in $E$ such that for some gate $v$ in $D$, 
the variables $x$ and $x'$ are, respectively, inside and outside the subcircuit of $D$ rooted at $v$.  
In this case, as we now show, all certificates of $D$ through $v$ set $x$ to $1$, 
or all certificates of $D$ through $v$ set $x'$ to $1$.

\begin{theorem}\label{th:treesagreeonedges}
Let $D$ be a nice DNNF, $v \in D$, 
$x \in \mathsf{vars}(\mathsf{sub}(D,v))$,  
and $x' \in \mathsf{vars}(D) \setminus \mathsf{vars}(\mathsf{sub}(D,v))$.  %\marginpar{$\mathsf{vars}(D)$ instof $V$}
If $\{x,x'\} \cap \mathsf{vars}(T) \neq \emptyset$ for all $T \in \mathsf{cert}(D)$, 
then at least one of the following two statements holds:
\begin{itemize}
\item $x \in \mathsf{vars}(T)$ for all $T \in \mathsf{cert}(D)$ such that $v \in T$.
\item $x' \in \mathsf{vars}(T)$ for all $T \in \mathsf{cert}(D)$ such that $v \in T$.
\end{itemize}
\end{theorem}
\begin{proof}
Let $\{T,T'\} \subseteq \mathsf{cert}(D)$ be such that $v \in T \cap T'$.  
As $D$ is constant free, by Proposition~\ref{prop:cert-are-trees}, 
the underlying graphs of the certificates of $D$ are trees.  By hypothesis, $\{x,x'\} \cap \mathsf{vars}(T) \neq \emptyset$ 
and $\{x,x'\} \cap \mathsf{vars}(T') \neq \emptyset$.   
We want to show that $x \in \mathsf{vars}(T) \cap \mathsf{vars}(T')$ 
or $x' \in \mathsf{vars}(T) \cap \mathsf{vars}(T')$.  

Assume towards a contradiction that $x \in \mathsf{vars}(T) \setminus \mathsf{vars}(T')$ 
and $x' \in \mathsf{vars}(T') \setminus \mathsf{vars}(T)$; 
the case where $x' \in \mathsf{vars}(T) \setminus \mathsf{vars}(T')$ 
and $x \in \mathsf{vars}(T') \setminus \mathsf{vars}(T)$ is symmetric.  

First, we observe that $x \not\in \mathsf{vars}(T) \setminus \mathsf{vars}(\mathsf{sub}(T,v))$ since, by Proposition~\ref{prop:cert-are-trees}, the leaves of $T$ are labeled with distinct variables and $x \in \mathsf{vars}(\mathsf{sub}(T,v))$.

Second, since $\mathsf{vars}(\mathsf{sub}(T',v)) \subseteq \mathsf{vars}(\mathsf{sub}(D,v))$, 
and $x' \not\in \mathsf{vars}(\mathsf{sub}(D,v))$ by hypothesis, 
it holds that $x' \not\in \mathsf{vars}(\mathsf{sub}(T',v))$.  
Therefore, $$\{x,x'\} \cap \mathsf{vars}(T) \setminus \mathsf{vars}(\mathsf{sub}(T,v))=\emptyset \textup{ and } \{x,x'\} \cap \mathsf{vars}(\mathsf{sub}(T',v))=\emptyset\text{.}$$  Now, 
the tree $S$ obtained by replacing in $T$ the subtree rooted at $v$ by the subtree rooted at $v$ in $T'$ 
is a certificate of $D$; moreover, $\{x,x'\} \cap \mathsf{vars}(S)=\emptyset$, 
contradicting the hypothesis that all certificates of $D$ have a nonempty intersection with $\{x,x'\}$.
\end{proof}

Therefore, if $G$ contains a matching $M$ such that each edge in the matching 
satisfies the condition of the previous statement, namely 
there is a gate $v$ in $D$ such that each edge in $M$ has 
one vertex inside and the other vertex outside the subcircuit of $D$ rooted at $v$, 
then all certificates of $D$ through $v$ agree on setting $|M|$ variables to $1$.

\begin{corollary}\label{cor:large-point-large-matching}
Let $E$ be a graph CNF whose underlying graph is $G=(V,E)$ 
and let $D$ be a nice DNNF equivalent to $E$.  Let $v$ be a gate in $D$ 
and $M$ be a matching in $G$ between $\vars(\sub(D,v))$ and $\vars(D) \setminus \vars(\sub(D,v))$. There exists $I_v \subseteq V$ such that $|I_v| = |M|$ and for all $T \in \cert(D)$, if $v \in T$ then $I_v \subseteq \vars(T)$.
\end{corollary}
\begin{proof}
  Let $e = \{x,x'\} \in M$ with $x \in \vars(\sub(D,v))$ and $x' \in \vars(D) \setminus \vars(\sub(D,v))$. 
Since $x \vee x'$ is a clause in the CNF $E$, for all $T \in \cert(D)$, either $x \in \vars(T)$ or $x' \in \vars(T)$. 
Thus by Theorem~\ref{th:treesagreeonedges}, either $x \in \vars(T)$ for all $T \in \cert(D)$ such that $v \in T$ or $x' \in \vars(T)$ for all $T \in \cert(D)$ such that $v \in T$. Let $x_e$ be the vertex of $e$ that is in every $T \in \cert(D)$ such that $v \in T$. We choose $I_v = \{x_e \mid e \in M\}$.

  By construction, it is clear that for all $T \in \cert(D)$ such that $v \in T$, we have $I_v \subseteq \vars(T)$. Moreover, 
since $M$ is a matching, for $e,e' \in M$, if $e \neq e'$ then $e \cap e' = \emptyset$ and thus $x_e \neq x_{e'}$, that is $|I_v| = |M|$.
\end{proof}

\subsection{Bottleneck Argument}\label{sect:bottleneck}

We are now ready to set up our bottleneck argument.  In the sequel, 
$D$ is a nice DNNF computing a graph CNF $E$ whose underlying graph is 
an expander $G=(V,E)$.  We define a subset $B$ of gates of $D$ (Lemma~\ref{lemma:gates}) 
and, for each gate $v\in B$, a subset $I_v$ of $V$ (Lemma~\ref{lemma:point}) 
in such a way that the fraction of vertex covers containing $I_v$ 
is exponentially small in $|V|$ (Corollary~\ref{cor:keyrazgon}), 
hence $B$ is exponentially large in $|V|$ (Theorem~\ref{th:maintechnical}).  

\subsubsection{Finding the Bottleneck Gates.}

We define the bottleneck $B \subseteq D$ as follows.  
For every certificate $T$ of $D$ we find (in a greedy fashion) a node $v_T$ in $T$ 
such that the subcircuit of $D$ rooted at $v_T$ has a large but not too large number of variables, 
and we put $v_T$ into $B$.

\begin{lemma}\label{lemma:gate}
Let $E$ be a graph CNF whose underlying graph $G=(V,E)$ is connected and has degree $d$ ($d \geq 3$).  
Let $D$ be a nice DNNF equivalent to $E$ and let $T \in \mathsf{cert}(D)$.  
There exists a gate $v_T \in T$ 
such that 
\begin{equation}\label{eq:largepoint}
|V|/(d+1) \leq |\mathsf{vars}(\mathsf{sub}(D,v_T))| \leq |V|/2\text{.} 
\end{equation}
\end{lemma}
\begin{proof}
We claim that $\frac{|V|}{(d+1)} \leq |\mathsf{vars}(T)|\text{.}$
Indeed, let $f$ be the assignment defined by $f(v)=1$ if and only if $v \in \mathsf{vars}(T)$.  
Then $f$ satisfies $D$ by Proposition~\ref{prop:mods-vs-trees}.  
Since $D$ computes $E$, we have that $\vars(T)=\{ v \colon f(v)=1 \}$ is a vertex cover of $G$. Thus by Proposition~\ref{prop:vcbounds}, $|\vars(T)| \geq |V|/(d+1)$.  

By Proposition~\ref{prop:tree}, with $\ell = |V|/(d+1)$, there exists a vertex $v_T$ in $T$ such that $|V|/(d+1) \leq |\mathsf{vars}(\mathsf{sub}(D,v_T))| \leq 2|V|/(d+1) \leq |V|/2$ where the last inequality comes from the fact that $d+1 \geq 4$.
\end{proof}

\begin{lemma}\label{lemma:gates}
Let $E$ be a graph CNF whose underlying graph $G=(V,E)$ is connected and has degree $d$ ($d \geq 3$).  
Let $D$ be a nice DNNF equivalent to $E$.  There exist $B \subseteq D$ such that:
\begin{enumerate}
\item $|V|/(d+1) \leq |\mathsf{vars}(\mathsf{sub}(D,v))| \leq |V|/2$, for all $v \in B$.
\item For all $T \in \mathsf{cert}(D)$ there exists $v \in B$ such that $v \in T$.
\end{enumerate}
\end{lemma}

\begin{proof}
We simply choose $B = \{v_T \mid T \in \cert(D)\}$ where $v_T$ is the vertex of $T$ from Lemma~\ref{lemma:gate}.
\end{proof}

\subsubsection{Mapping the Vertex Covers.} For each $v \in B$, we find a large matching in $G$ between variables inside and outside 
the subcircuit rooted at $v$, 
and then use Corollary~\ref{cor:large-point-large-matching} to derive 
a large set $I_v \subseteq V$ such that for all certificates $T$ through $v$ it holds that $I_v \subseteq \mathsf{vars}(T)$.

Recall that a \emph{matching} in a graph $G$ is a subset $M$ of the edges 
such that $\{u,v\} \cap \{u',v'\} \neq \emptyset$ 
for every two distinct edges $\{u,v\}$ and $\{u',v'\}$ in $M$.  
For disjoint subsets $V'$ and $V''$ of the vertices of $G$, 
a matching $M$ in $G$ is said \emph{between $V'$ and $V''$} 
if every edge in $M$ intersects both $V'$ and $V''$.  

\begin{lemma}\label{lemma:point}
Let $E$ be a graph CNF whose underlying graph $G=(V,E)$ is a $(c,d)$-expander ($d \geq 3$, $c>0$), 
let $D$ be a nice DNNF equivalent to $E$, 
and let $v \in D$ such that $|V|/(d+1) \leq \vars(\sub(D,v)) \leq |V|/2$. There exists $I_v \subseteq V$ such that:
\begin{enumerate}
\item For all $T \in \mathsf{cert}(D)$ such that $v \in
  T$ it holds that $I_v \subseteq \mathsf{vars}(T)$.
\item $|I_v| \geq c|V|/(2 d^2)$.
\end{enumerate}
\end{lemma}
\begin{proof}
  The idea is to construct a matching between $S = \vars(\sub(D,v))$ and $V \setminus S$ of size at least 
$c|V|/(2 d^2)$ and then apply Corollary~\ref{cor:large-point-large-matching}. 
Since $|V|/(d+1) \leq |S| \leq |V|/2$ and $G$ is a $(c,d)$-expander, by (\ref{eq:expansion}) and (\ref{eq:largepoint}) 
we have that 
$$|\neigh_S| \geq c|S| \geq c|V|/(d+1)\text{.}$$
We construct a matching $M$ between between $S$ 
and $V \setminus S$ in $G$ as follows.  
Pick an edge $\{v,w\} \in E$ with $v \in S$ and $w \in \neigh_S \subseteq V \setminus S$; 
add $\{v,w\}$ to $M$; 
delete $v$ from $S$, 
$w$ from $\neigh_S$, 
the vertices in $S$ with no neighbors in $\neigh_S$ after the deletion of $w$, 
and the vertices in $\neigh_S$ with no neighbors in $S$ after the deletion of $v$; 
iterate on the updated $S$ and $\neigh_S$, until either $S=\emptyset$ or 
$\neigh_S=\emptyset$.  At each step, we delete 
at most $d$ vertices in $S$ and 
at most $d$ vertices in $\neigh_S$.  

Hence, we iterate for at least 
$$s \geq \min \left\{ \frac{|S|}{d}, \frac{|\neigh_S|}{d} \right\} \geq 
\frac{\min\{1,c\}}{d(d+1)}|V| \geq \frac{c|V|}{2 d^2}$$
steps ($d \geq 3$, $c \leq 1$).  So we have that 
$|M| \geq s \geq c|V|/(2 d^2)$, 
and we are done. Now, applying Corollary~\ref{cor:large-point-large-matching} on $v$ and the matching $M$ yields the result.
\end{proof}

\subsubsection{Proving the Lower Bound.}  We conclude proving 
that for every $v \in B$ the fraction of vertex covers of $G$ containing $I_v$ 
is exponentially small (Corollary~\ref{cor:keyrazgon}).  
On the other hand, by construction, every vertex cover of $G$ 
contains a set $I_v$ for some $v \in B$, so that the union (over $v \in B$) 
of the vertex covers of $G$ containing $I_v$ coincides with the vertex covers of $G$; 
hence $B$ is exponentially large (Theorem~\ref{th:maintechnical}).  

Let $G = (V,E)$ be a graph and let $S \subseteq V$.  
We denote by $\VC(G,S)$ the set of vertex covers of $G$ containing~$S$.

\begin{theorem}
\label{th:florent}
Let $G = (V,E)$ be a graph, $S \subseteq V$, and $s \in S$.  Then,
$${|\VC(G, S)|} \leq \left (\frac{2^{d(s)}}{1+2^{d(s)}}\right){|\VC(G,S\setminus \{s\})|}\text{.}$$
\end{theorem}

\pfkeyflorent

\begin{corollary}
\label{cor:keyrazgon}
Let $G = (V,E)$ be a graph of degree $d$ and let $S \subseteq V$.  Then, 
$$|\VC(G, S)| \leq \left(\frac{2^{d}}{1+2^{d}}\right)^{|S|}|\VC(G)| \text{.}$$
\end{corollary}

\pfflorent

\begin{theorem}
\label{th:maintechnical}
Let $G=(V,E)$ be a $(c,d)$-expander such that $|V|\geq 2$. Then
$$\mathsf{DNNF}(E) \geq 2^{g(c,d)\mathsf{size}(E)-1}\text{,}$$
where $g(c,d)=\frac{c \cdot f(d)}{6 d^3}$ and $f(d)=\log_2(1+2^{-d})>0$.
\end{theorem}
\begin{proof}
Let $D$ be  a nice DNNF equivalent to $E$. By Proposition~\ref{prop:normaldnnf}, we can assume that $\mathsf{size}(D) \leq 2 \cdot\mathsf{DNNF}(E)$.   

Let $B \subseteq D$ be the set of gates from Lemma~\ref{lemma:gates}.  Let $v \in B$. By construction, $|V|/(d+1) \leq |\vars(\sub(D,v))| \leq |V|/2$. Thus by Lemma~\ref{lemma:point}, there exists $I_v$ such that for all $T \in \cert(D)$ such that $v \in T$, we have $I_v \subseteq \vars(T)$ 
and $|I_v| \geq h(c,d)|V|$ where $h(c,d)=c/(2 d^2)$.  By Corollary~\ref{cor:keyrazgon}, 
$$|\VC(G,I_v)| \leq \left(\frac{2^d}{1+2^d}\right)^{|I_v|}|\VC(G)|\text{.}$$
As $f(d)=-\log_2(\frac{2^d}{1+2^d})$, we have that, for all $v \in B$,
$$|\VC(G,I_v)| \leq 2^{-f(d)h(c,d)|V|}|\VC(G)|\text{.}$$

We claim that  
$\VC(G) = \bigcup_{v \in B}\VC(G,I_v)$. For the nontrivial containment, let $C \in \VC(G)$ and, 
by Proposition~\ref{prop:mods-vs-trees}, 
let $T \in \mathsf{cert}(D)$ be such that $\mathsf{vars}(T) \subseteq C$.  
By Lemma~\ref{lemma:gates}$(ii)$, let $v \in B$ be such that $v \in T$.  
Then $I_v \subseteq \mathsf{vars}(T)$ by Lemma~\ref{lemma:point}$(i)$, 
so that $I_v \subseteq C$, that is, $C \in \VC(G,I_v)$.  Therefore,
$$|\VC(G)| \leq \sum_{v \in B}|\VC(G,I_v)| \leq 2^{-f(d)h(c,d)|V|}|\VC(G)| \cdot |B|\text{,}$$
from which $|B| \geq 2^{f(d)h(c,d)|V|}$.

Now observe that $|E| \leq d|V|$, because $G$ has degree $d$.  Thus, 
the CNF $E$ has at most $d|V|$ clauses, each of at most $2$ literals, so that $d|V|+2d|V|=3d|V| \geq \mathsf{size}(E)$. 
Since $\mathsf{size}(D) \geq |B|$ and $g(c,d) = \frac{f(d)h(c,d)}{3d}$, we finally have
$$\mathsf{DNNF}(E) \geq \mathsf{size}(D)/2 \geq 2^{g(c,d)\mathsf{size}(E)-1}\text{.}$$\end{proof}

Our main result follows from the previous theorem.

\begin{theorem}
\label{thm:mainnontechnical}
There exist a class $\mathcal{C}$ of CNF formulas and a constant $c > 0$ such that $\mathsf{DNNF}(F) \geq 2^{c \cdot \mathsf{size}(F)}$ for each formula $F \in \mathcal{C}$.  Indeed, $\mathcal{C}$ is a class of read $3$ times monotone $2$-CNFs.
\end{theorem}

\newcommand{\pfmainnontechnical}[0]{
\begin{proof}
By Theorem~\ref{th:expex}, there exists a family $\mathcal{G}=\{G_i=(V_i,E_i) \colon i \in \mathbb{N}\}$ of $(e,3)$-expander 
graphs such that $|V_i| \geq 2$ for all $i \in \mathbb{N}$ and $|V_i| \to \infty$ as $i \to \infty$ ($e>0$).  Every graph 
in $\mathcal{G}$ is connected; in particular, it does not contain isolated vertices.  Therefore $E_i$ is a CNF 
for every $i \in \mathbb{N}$, and indeed $E_i$ is a read $3$ times monotone $2$-CNF.  

Since $|V_i| \to \infty$ as $i \to \infty$ and each graph in $\mathcal{G}$ satisfies (\ref{eq:expansion}), 
there exists an infinite subset $I \subseteq \mathbb{N}$ such that $\mathsf{size}(E_i)< \mathsf{size}(E_{i+1})$ for all $i \in I$.  
Choose $c>0$ and $j \in I$ large enough such that $g(e,3) \cdot \size{E_{j}}-1 \geq c \cdot \size{E_{j}}$, 
where $g(\cdot,\cdot)$ is as in the statement of Theorem~\ref{th:maintechnical}.  
It follows from Theorem~\ref{th:maintechnical} that 
$\mathsf{DNNF}(E_{j}) \geq 2^{g(e,3) \cdot \size{E_{j}}-1} \geq 2^{c \cdot \size{E_{j}}}\text{;}$
we take $\mathcal{C}=\{ E_i \colon i \in I \text{ and } i \geq j \}$, and the statement is proved.  
\end{proof}}

\pfmainnontechnical

\newcommand{\corollariessection}{
\section{Corollaries}\label{sect:cors}
 
In this section we will prove the corollaries of Theorem~\ref{thm:mainnontechnical} we sketched in the introduction. 

\medskip \noindent Let $F$ be a CNF. We say that a clause $C$ is \emph{entailed} by $F$ if every 
satisfying assignment of $F$ also satisfies $C$. We say that a clause $C'$ \emph{subsumes} $C$, if $C'\subseteq C$. 
A CNF $F$ is in prime implicates form (short PI) if every clause that is entailed by $F$ is subsumed by 
a clause that appears in $F$ and no clause in $F$ is subsumed by another.  Note that CNFs in PI form can express all Boolean functions but it is known that encoding in 
PI form may generally be exponentially bigger than general CNF~\cite{DarwicheM02}. 

\newcommand{\pfpi}[0]{
\begin{proof}
 Let $F$ be a monotone 2-CNF formula. We first note that trivially no clause in a 2-CNF subsumes another.
 
 Now let $C$ be a clause entailed by $F$ and assume by way of contradiction that $C$ is not subsumed by any clause of $F$, that is, 
every clause in $F$ contains a positive literal not in $C$. Let $C'$ be the clause we get from $C$ by deleting all negative literals. We claim that $C'$ is entailed by $F$. To see this, consider a satisfying assignment~$f$ of~$F$. Let $f'$ be the assignment we get from $f$ by setting the variables that are negated in $C$ to $1$. Since $F$ is monotone, this is still a satisfying assignment of $F$ and thus of $C$. Consequently, $C$ is satisfied by one of its positive literals in $f'$ and thus in $f$. Thus $f$ satisfies $C'$ and it follows that $F$ entails $C'$.
 
 Now let $f$ be the assignment that sets all variables in $C'$ to $0$ and all other variables to $1$. Since $C$ and thus also $C'$ is not subsumed by any clause of $F$, the assignment $f$ satisfies $F$. But by construction $f$ does not satisfy $C'$ which is a contradiction.
\end{proof}}

\begin{lemma}\label{lem:PI}
Every monotone 2-CNF formula is in PI form.
\end{lemma}

\pfpi

Remember that the formulas of Theorem~\ref{thm:mainnontechnical} are monotone 2-CNF formulas.  
We directly get the promised separation from Lemma~\ref{lem:PI} and Theorem~\ref{thm:mainnontechnical}.

\begin{corollary}\label{cor:pi}
There exist a class $\mathcal{C}$ of CNFs in PI form and a constant $c>0$ such that, 
for every formula $F$ in $\mathcal{C}$, every DNNF equivalent to $F$ has size at least $2^{c \cdot \size{F}}$.
\end{corollary}
It follows that $\mathrm{L_1}$ can be exponentially more succinct than $\mathrm{L_2}$ for any two representation languages $\mathrm{L_1} \supseteq \mathrm{PI}$ and $\mathrm{L_2} \subseteq \mathrm{DNNF}$. In particular, this holds if $\mathrm{L_1} \in \{\mathrm{PI}, \mathrm{CNF}, \mathrm{NNF}\}$ and $\mathrm{L_2} \in \{\mathrm{d}\textup{-}\mathrm{DNNF}, \mathrm{DNNF}\}$. Here, $\mathrm{d}\textup{-}\mathrm{DNNF}$ denotes the language of \emph{deterministic} DNNFs, that is, DNNFs where subcircuits leading into a $\lor$-gate never simultaneously evaluate to $1$. This answers several questions concerning the relative succinctness of common representation languages~\cite{DarwicheM02}.

We also observe that DNNFs are not closed under negation.

\newcommand{\pfdnfs}[0]{
\begin{proof}
Let $\mathcal{C}$ be the class of 2-CNFs from Theorem~\ref{thm:mainnontechnical}. Let $\mathcal{D}$ be the class of 2-DNFs 
we get by negating the formulas in $\mathcal{C}$. Now negating $\mathcal{D}$ gives the class $\mathcal{C}$ 
again, for which we have the lower bound from Theorem~\ref{thm:mainnontechnical}.
\end{proof}}

\begin{lemma}\label{lemma:dnfs}
There exist a class $\mathcal{D}$ of 2-DNF formulas and a constant $c>0$ such that, 
for every formula $D$ in $\mathcal{D}$, every DNNF equivalent to $\neg D$ has size at least $2^{c \cdot \size{D}}$.  
\end{lemma}

\pfdnfs

Observing that DNF is a restricted form of DNNF, we get the following non-closure result which was only known conditionally before.

\begin{corollary}\label{cor:negation}
There exist a class $\mathcal{D}$ of DNNFs and a constant $c>0$ such that, 
for every formula $D$ in $\mathcal{D}$, every DNNF equivalent to $\neg D$ has size at least $2^{c \cdot \size{D}}$.  
\end{corollary}
}

\corollariessection

\section{Conclusion}

We proved an unconditional, strongly exponential separation between the representational power of CNFs and that of DNNFs and discussed its consequences in the  area of knowledge compilation~\cite{DarwicheM02}. Let us close by mentioning directions for future research.

In order to prove the lower bound of Theorem~\ref{thm:mainnontechnical}, we generalized arguments concerning paths in branching programs to the tree-shaped certificates of DNNFs. It would be interesting to know whether other lower bounds for branching programs can be lifted to (suitably restricted) versions of DNNFs along similar lines.

Recent progress notwithstanding~\cite{BeameLRS13}, several separations between well-known representation languages are known to hold only conditionally~\cite{DarwicheM02}. For instance, it is known that DNFs cannot be compiled efficiently into so-called deterministic DNNFs unless the polynomial hierarchy collapses~\cite{Selman96,Cadoli02}.  
It would be interesting to show this separation unconditionally.  

Finally, there is a long line of research proving upper bounds for DNNFs and restrictions (see \cite{Razgon13,OztokD14b,OztokD14} for some recent contributions). We believe that these results should be complemented by lower bounds as in~\cite{Razgon13,Razgon14c}, and hope that the ideas developed in this paper will contribute to this project.

\paragraph*{Acknowledgments.} The first author was supported by the European Research Council (Complex Reason, 239962)
and the FWF Austrian Science Fund (Parameterized Compilation, P26200).
The third author has received partial support by a Qualcomm grant administered by \'Ecole Polytechnique. 
The results of this paper were conceived during a research stay of the second and third author at the Vienna University of Technology. The stay of the second author was made possible by financial support by the ANR Blanc COMPA. The stay of the third author was made possible by financial support by the ANR Blanc International ALCOCLAN. The fourth author was supported by the European Research Council (Complex Reason, 239962).

\end{document}